\documentclass{svmult} 
  
\usepackage{mathptmx}   
\usepackage{helvet} 
\usepackage{courier}
\usepackage{amsmath,amsfonts}  
\usepackage{graphicx}
\usepackage[bottom]{footmisc}           
    
\usepackage{mathptmx}         
\usepackage{helvet}            
\usepackage{courier}            
\usepackage{type1cm}             
\usepackage{makeidx}             
\usepackage{graphicx}          
\usepackage{multicol}          
\usepackage[bottom]{footmisc}  
     
\usepackage{epsfig}     
\usepackage{psfrag,rotating}     
\usepackage{amssymb,floatflt,enumerate} 
\usepackage{amsmath,amscd,psfrag,leqno} 
\usepackage[mathscr]{eucal}  
\usepackage{shadethm}           
\usepackage{graphicx,subfigure}   
\usepackage{overpic,contour}

\contourlength{0.3mm}  
 
\def\Beweisende{\square}            
\def\BewEnde{\hfill{\Beweisende}}

\def\phm{{\hphantom{-}}}

\def\phi{\varphi}

\def\RR{{\mathbb R}}

\def\EE{{\mathbb E}}

\def\overset#1#2{\mathop{#2}\limits^{#1}}
\def\Vkt#1{{\mathbf #1}}

\definecolor{blau}{rgb}{0,0,1}
\newcommand{\blau}{\color{blau}}
\definecolor{rot}{rgb}{1,0,0}
\newcommand{\rot}{\color{rot}}
\definecolor{green}{rgb}{0,1,0} 
\newcommand{\green}{\color{green}}
\definecolor{cyan}{rgb}{0,1,1}
\newcommand{\cyan}{\color{cyan}}
\definecolor{violet}{rgb}{1,.54,1}
\newcommand{\violet}{\color{violet}}
\definecolor{magenta}{rgb}{1,0,1}
\newcommand{\magenta}{\color{magenta}}

\makeindex
\begin{document}

\title*{Evaluating the %snapping capability 
snappability of bar-joint frameworks}
\author{G. Nawratil}
\authorrunning{G. Nawratil}
\institute{
  Institute of Discrete Mathematics and Geometry \& Center for Geometry and Computational Design, TU Wien, Austria, 
  \email{nawratil@geometrie.tuwien.ac.at}}

\maketitle

\abstract{
It is well-known that there exist bar-joint frameworks (without continuous flexions) 
whose physical models can snap between different realizations due to non-destructive elastic deformations of material. 
We present a method to measure these 
snapping capability -- shortly called {\it snappability} -- based on the 
total elastic strain energy of the framework by computing the deformation of all bars using Hooke's law. 
The presented theoretical results give further connections between 
shakiness and snapping beside the well-known technique of averaging and deaveraging. 
} 


\section{Introduction}
\label{sec:introduction}

We consider a framework in the Euclidean $n$-dimensional space $\EE^n$ which consists of a knot set $\mathcal{K}=\left\{K_1,K_2,\ldots,K_s\right\}$ 
and an abstract graph $G$ on $\mathcal{K}$ fixing the combinatorial structure. We denote the edge connecting  $K_i$ to $K_j$  by
$e_{ij}$ with $i<j$ and collect all indices of knots edge-connected to $K_i$ in the knot neighborhood $N_i$.  
Moreover we denote the number of edges in the graph by $b$ and assign\footnote{This assignment corresponds to the definition of the intrinsic metric 
of the framework.} a length $L_{ij}\in\RR_{>0}$ to each edge $e_{ij}$. In general this assignment does not determine the shape of the framework uniquely 
thus a framework has different incongruent  
realizations. For example, a triangular framework has in general two realizations in $\EE^2$, 
which are not congruent with respect to the group of direct isometries. 
If we consider the isometry group
then this number halves. \\
By materializing all edges by straight bars and linking them by $S_n$-joints\footnote{$S_n$ denotes the 
spherical joint, which enables the group of spherical motions SO($n$) of $\EE^n$. Note that a $S_2$-joint equals 
a rotational joint (R-joint).} in the knots, 
we end up with a so-called bar-joint framework. We assume that (1) all bars are uniform made of a
homogeneous isotropic material deforming at constant volume (i.e.\ Poisson's ratio $\nu=1/2$)
and that (2) all $S_n$-joints are without clearance. 

In the following we consider the configuration of knots $\Vkt k:=(\Vkt k_1,\ldots ,\Vkt k_s)$, where $\Vkt k_i$ denotes the $n$-dimensional coordinate vector 
of the knot $K_i$, 
which together with  the graph $G$ implies the framework's realization $G(\Vkt k)$.
In the rigidity community (e.g.\ \cite{connelly_book}) each edge $e_{ij}$ is assigned with a {\it stress (coefficient)} ${\omega}_{ij}\in\RR$.  
If in every knot $i\in\left\{1,\ldots, s\right\}$ the so-called {\it equilibrium condition}
\begin{equation}\label{equilibrium}
\sum_{i<j\in N_i}{\omega}_{ij}(\Vkt k_i-\Vkt k_j) + \sum_{i>j\in N_i}{\omega}_{ji}(\Vkt k_i-\Vkt k_j)=\Vkt o
\end{equation}
is fulfilled, where $\Vkt o$ denotes the $n$-dimensional zero-vector, then the $b$-dimensional vector ${\omega}=(\ldots, {\omega}_{ij} , \ldots)^T$ is refereed as 
{\it self-stress} (or {\it equilibrium stress}).  
According to Gluck \cite{gluck} and Roth \cite{roth} the existence of a non-zero self-stress\footnote{$\omega$ differs from the $b$-dimensional zero vector.}
corresponds to the infinitesimal flexibility (shakiness) of the framework's realization $G(\Vkt k)$.  

Shakiness (of order one\footnote{Each additional coinciding realization 
raises the order of the infinitesimal flexibility by one \cite{wohlhart}.}) can also be seen as the limiting case where two realizations of a framework coincide 
\cite{stachel_wunderlich,stachel_between,wohlhart}. 
In contrast a realization $G(\Vkt k')$ is called a snapping realization if it is {\it close enough} 
to another incongruent realization such that the physical model can snap (flip/jump) into this realization due to non-destructive elastic deformations of material. 
The open problem in this context is the meaning of {\it closeness}, which is tackled in this article. In more detail, 
we present a method to measure the snapping capability (shortly called {\it snappability}) of a realization $G(\Vkt k')$, based on  
the total elastic strain energy of the \bigskip framework.

Before we  plunge in medias res we provide a short review on snapping (also called {\it multistable}; cf.\ \cite{goldberg}) 
structures. 
During the last years the interest in these structures has increased  due to practical applications (e.g.\ \cite{rafsanjani,haghpanah,shang}).

It is pointed out in \cite{stachel_between} that there is a direct connection between shakiness and snapping
through the technique of {\it averaging} and {\it deaveraging}, respectively (cf.\ \cite[page 1604]{schulze} and 
\cite{ivan}).
The latter allows to construct snapping frameworks in any dimension. 
Moreover for snapping bipartite frameworks in $\EE^n$ an explicit result in terms of confocal hyperquadrics is known 
(cf.\ \cite[page 112]{stachel_between} under consideration of \cite{stachel_palermo}). 
Most results are known for the dimension $n=3$, which are summarized next.

There is a series of papers of Walter Wunderlich on snapping spatial structures 
(octaeder, Bennett mechanisms, antiprisms, icosaeder, dodecaeder), 
which are reviewed in \cite{stachel_wunderlich}. In this context also the paper \cite{goldberg} 
should be cited, where {\it buckling polyhedral surfaces} and {\it Siamese dipyramids} are introduced. 
Snapping structures are also related to so-called {\it model flexors}\footnote{Mathematically these structures do not posses a 
continuous flexibility but due to free bendings without visible distortions of materials their physical models flex.} 
(cf.\ \cite{milka}) as in some cases the model flexibility can be reasoned by the snapping through different realizations. 
Examples for this phenomenon are the so-called {\it four-horn} \cite{schwabe} or the already mentioned {\it Siamese dipyramids}.
The latter are studied in more detail in \cite{gorkavyy}, especially how relative variations on the edge lengths 
produce significant relative variations in the spatial shape. The authors of \cite{gorkavyy} also suggested estimates to quantify these 
intrinsic and extrinsic variations.

\section{Physical model of deformation}

First we consider a single bar $e_{ij}$ and apply equal but opposite directed forces $F_{ij}$ to it's ends pointing outwards/inwards the bar, 
which imply a tensile/compression stress leading to an expansion/decrease of the bar. 
According to Hooke's law, which can be applied due to the elastic deformation during the process of snapping, 
 the tensile/compression stress $\delta_{ij}$ in a uniform bar equals the product of the modulus of elasticity \footnote{In this paper we assume $E_{ij}>0$ as for conventional structural material $E_{ij}$ is positive.} $E_{ij}>0$ 
and the Cauchy/engineering strain\footnote{Equals in this case the engineering normal strain $\varepsilon_{ij}$ of a material line segment $e_{ij}$.}  
$\varepsilon_{ij}$; i.e. 
\begin{equation}\label{eq2}
\delta_{ij}=E_{ij}\varepsilon_{ij} \quad \text{with} \quad \varepsilon_{ij}=\frac{l_{ij}-L_{ij}}{L_{ij}} \quad \text{for} \quad i<j
\end{equation}
where $l_{ij}\in\RR_{>0}$ is the deformed length (stressed length) of the bar $e_{ij}$, while 
$L_{ij}$ is the original length (unstressed length) of the bar $e_{ij}$. 
Note that  $\delta_{ij}>0$  corresponds to a stretching and $\delta_{ij}<0$ to a compression. 
Taking Eq.\ (\ref{eq2}) into account the well-known relation $F_{ij}=\delta_{ij}A_{ij}$, where $A_{ij}$ is the cross-sectional area of the bar $e_{ij}$, can 
be rewritten as
\begin{equation}\label{kraft}
F_{ij}=E_{ij}A_{ij}\frac{l_{ij}-L_{ij}}{L_{ij}}  
\quad \text{for} \quad i<j. 
\end{equation} 
This force acts on the end points $K_i$ and $K_j$ of the involved bar $e_{ij}$  by the force vectors
\begin{equation}
\Vkt f_{ij}= F_{ij}\frac{\Vkt k_i-\Vkt k_j}{\| \Vkt k_i-\Vkt k_j \|} \quad \text{and} \quad 
\Vkt f_{ji}= -\Vkt f_{ij},
\end{equation}
where $\|.\|$ denotes the {\it standard norm}. 
Moreover, the elastic strain energy\footnote{Note that elastic strain energy is a form of potential energy.} stored in the bar $e_{ij}$  can be computed by 
\begin{equation}
U_{ij}=\tfrac{1}{2}F_{ij}(l_{ij}-L_{ij}) \quad  \overset{\text{$(\ref{kraft})\,\,$}}{\implies} \quad
U_{ij}=\frac{1}{2}\frac{E_{ij}A_{ij}}{L_{ij}} (l_{ij}-L_{ij})^2
\end{equation}
according to \cite[page 512]{mittemeijer}, where  $E_{ij}A_{ij}/L_{ij}$ is the axial stiffness  
of the bar $e_{ij}$. 
As a consequence the total deformation energy $U$ of the framework reads as 
\begin{equation}\label{energie}
U:=\sum_{i<j}U_{ij}=\frac{1}{2}\sum_{i<j}\frac{E_{ij}A_{ij}}{L_{ij}}(l_{ij}-L_{ij})^2.
\end{equation}

\begin{remark}
In the study of rigidity and stability of frameworks similar functions appear; e.g.\ the function 
\begin{equation}
\Phi=\frac{1}{2}\sum_{i<j} \frac{F_{ij}}{L_{ij}}\left(
l_{ij}^2-L_{ij}^2
\right)
\end{equation}
was used by K\"otter \cite{koetter}. 
Connelly and Whiteley \cite{CW1992} based their results on a potential energy function 
$\sum_{ij}H_{ij}\|\Vkt k_i-\Vkt k_j\|^2$
where $H_{ij}$ is any real-valued function of one variable, which has a local minimum at $\|\Vkt k_i-\Vkt k_j\|^2$. \hfill $\diamond$
\end{remark}

Expressing $l_{ij}$ in dependence of the knots $K_i$ and $K_j$ the system of 
partial derivatives 
\begin{equation}
\frac{\partial U}{\partial k_{i,1}}=0, \quad
\frac{\partial U}{\partial k_{i,2}}=0, \quad
 \ldots ,\quad 
\frac{\partial U}{\partial k_{i,n}}=0,
\end{equation}
where $(k_{i,1},\ldots , k_{i,n})$ is the coordinate vector of $\Vkt k_i$,  
equals the condition $\sum_{j\in N_i} \Vkt f_{ij}=\Vkt o$ for $i=1,\ldots, s$.
With respect to our physical model the meaning of the real values ${\omega}_{ij}$ in Eq.\ (\ref{equilibrium}) is 
$F_{ij}/l_{ij}$, which is the so-called {\it force density} with respect to the stressed length \cite{tibert}.
Moreover it should be pointed out that each {\it critical point} of the total elastic strain-energy $U(\Vkt k)$ 
corresponds to a self-stressed framework realization $G(\Vkt k)$, which is called {\it deformed} for $U(\Vkt k)>0$ and {\it undeformed} for $U(\Vkt k)=0$.

\begin{remark}
The system of equations $\sum_{j\in N_i} \Vkt f_{ij}=\Vkt o$ was also obtained by Linkwitz and Schek \cite[page 149ff.]{schek}, 
who studied the form finding problem for cable networks. Within their proposed {\it force density method} \cite{tibert} values are assigned to 
the mentioned {\it force densities} rendering these equations linear in the knot coordinates. \hfill $\diamond$
\end{remark}

\subsection{Metric interpretation of $U$}\label{geom_inter}

This physical model implies in the $b$-dimensional  
space $\RR^b$ of bar lengths the following scalar product $\langle ., . \rangle_P$: $\RR^b\times\RR^b\rightarrow \RR$ with 
$(\Vkt x,\Vkt y)\mapsto\langle \Vkt x, \Vkt y \rangle_P:= \Vkt x^T P\, \Vkt y$ and
\begin{equation}
\Vkt x:=\begin{pmatrix} \vdots \\  x_{ij} \\ \vdots \end{pmatrix} , \quad
\Vkt y:=\begin{pmatrix} \vdots \\  y_{ij} \\ \vdots \end{pmatrix} , \quad
P:=\begin{pmatrix}
\ddots & \phm & \phm \\
\phm & \tfrac{E_{ij}A_{ij}}{2L_{ij}} & \phm \\
\phm & \phm & \ddots
\end{pmatrix},
\end{equation}
as the involved $b\times b$ diagonal matrix $P$ is positive definite. 
Consequently this scalar product induces a norm $\|.\|_P$: $\RR^b\rightarrow \RR_{\geq 0}$ with $\Vkt x\mapsto \|\Vkt x\|_P:=\sqrt{\langle \Vkt x, \Vkt x \rangle_P}$ 
and a metric (distance function) $d_P(.,.)$: $\RR^b\times\RR^b\rightarrow \RR_{\geq 0}$ with $(\Vkt x,\Vkt y)\mapsto d_P(\Vkt x,\Vkt y):=\|\Vkt x-\Vkt y\|_P$.
Thus $U$ of Eq.\ (\ref{energie}) can be seen as squared distance between the deformed framework $\Vkt l = (\ldots, l_{ij} ,\ldots)^T$ and 
the initial one $\Vkt L = (\ldots, L_{ij} ,\ldots)^T$; i.e.\ $U=d_P(\Vkt l,\Vkt L)^2$. 
It should be pointed out that this distance only depends on the intrinsic metric of the framework and is therefore independent of the actual realization. 

In order to reduce the  distance measure $d_P$ to its geometric core we set $E_{ij}=1$ and $A_{ij}=A$ for all bars 
in the remainder of the article.

\section{Snappability of realizations}

If we have a realization $G(\Vkt k)$ with non-zero self-stress and the corresponding 
total elastic energy $U(\Vkt k)$ is not at a local minimum, then small perturbations would deform the 
framework according to the {\it minimum total potential energy principle}.
Therefore we are interested in the set $\mathcal{S}$ of {\it stable} realizations, i.e.\ realizations which are at a local minimum of the total elastic strain-energy. 
Note that the set $\mathcal{S}$ is not empty as it contains at least the undeformed framework 
realizations.

\subsection{Computation of the set $\mathcal{S}$}\label{subsec:comp}

For the computation of the  set $\mathcal{S}$ for a given framework with a given intrinsic metric the following approach is used. 
We introduce new variables $q_{ij}$ fulfilling the side condition $\Lambda_{ij}=0$ with $\Lambda_{ij}:=q_{ij}^2-\|\Vkt k_i -\Vkt k_j\|^2$ 
and make the Lagrange ansatz 
\begin{equation}
F(\Vkt k,\Vkt q,\lambda):=\frac{A}{2}\sum_{i<j}\frac{1}{L_{ij}}(q_{ij}-L_{ij})^2 + \sum_{i<j}{\lambda}_{ij}\Lambda_{ij}
\end{equation}
with the $b$-dimensional vectors $\Vkt q:=(\ldots, q_{ij}, \ldots)^T$ and $\lambda:=(\ldots, {\lambda}_{ij}, \ldots)^T$, where the latter 
is composed of the Lagrange multipliers ${\lambda}_{ij}$. By taking the partial derivatives of $F$ with respect to the $(sn + 2b)$ 
variables $(\Vkt k,\Vkt q,\lambda)$ we obtain a system of equations, which is of algebraic nature. 
Therefore we can 
use homotopy continuation method (e.g.\ Bertini; cf.\ \cite{bates}), 
as other approaches (e.g.\ Gr\"obner base, resultant based elimination) are not promising due to the number of unknowns and degree of equations. 
First of all we can restrict to the obtained real critical points of $F$ with all $q_{ij}>0$ as only these correspond to realizations. 
This resulting set $\mathcal{R}$ of realizations is split into a set $\mathcal{M}$ and its absolute complement $\mathcal{M}^c=\mathcal{R}\setminus\mathcal{M}$, 
where the elements of $\mathcal{M}$ correspond to local minima of $U(\Vkt k)$. They can be identified by the so-called {\it second derivative test}; i.e.\ all eigenvalues of the  
{\it Hessian matrix} of the function $U(\Vkt k)$ 
are positive. 
Finally the desired set $\mathcal{S}$ can be obtained as the quotient $\mathcal{M}/$SE($n$), where SE($n$) denotes the group of direct isometries of $\EE^n$. 
In the same way we define the set $\mathcal{S}^c:=\mathcal{M}^c/$SE($n$) of {\it unstable} realizations, which is needed later on.

\subsection{Measuring the snappability}

The evaluation of the {\it snappability} has to be based on the intrinsic metric of the framework, as 
minor changes of this metric can heavily effect its spatial shape (cf.\ examples of the {\it four-horn} \cite{schwabe} 
and {\it Siamese dipyramids} \cite{gorkavyy}). 
Our intrinsic metric approach towards the determination of the {\it snappability} of an undeformed realization  
is based on the following theorem:
\begin{theorem}\label{thm1}
If a framework snaps out of a %given 
stable realization $G(\Vkt k')$ by applying the minimum energy needed to it, then 
the corresponding deformation of the realization has to pass a shaky realization $G(\Vkt k'')$ at the maximum state of deformation. 
\end{theorem}

\begin{proof}
We think of $U$ as a graph function over the space $\RR^{sn}$ of knots $\Vkt k$; i.e.\ the ordered pair $(\Vkt k, U(\Vkt k))$. 
In order to get out of the valley of the local minimum  $(\Vkt k', U(\Vkt k'))$, which corresponds to the given stable realization $G(\Vkt k')\in\mathcal{S}$, 
with a minimum of energy needed, one has to pass a {\it saddle point} $(\Vkt k'', U(\Vkt k''))$ of the graph, which corresponds to a realization $G(\Vkt k'')\in\mathcal{S}^c$. 
As all realizations of $\mathcal{S}^c$ are self-stressed and deformed 
($\Rightarrow$ no zero self-stress) they are infinitesimally flexible. \hfill $\BewEnde$
\end{proof}

Based on the pseudometric $d_S$: 
 $\RR^b\times\RR^b\rightarrow \RR_{\geq 0}$ with 
$(\Vkt l',\Vkt l'')\mapsto |d_p(\Vkt l'',\Vkt L)^2-d_p(\Vkt l',\Vkt L)^2|$ 
the  snappability index $s(\Vkt k')$ of  $G(\Vkt k')\in\mathcal{S}$ can be quantified as follows
\begin{equation}\label{eq:snap}
s(\Vkt k')=\frac{|U(\Vkt k'')-U(\Vkt k')|}{AL}=\frac{|U(\Vkt l'')-U(\Vkt l')|}{AL}=
\frac{d_S(\Vkt l',\Vkt l'')}{AL}
\end{equation}
where $\Vkt l' = (\ldots, l_{ij}' ,\ldots)^T\in \RR^b$ and $\Vkt l'' = (\ldots, l_{ij}'' ,\ldots)^T\in \RR^b$ are the bar lengths of the 
$G(\Vkt k')\in\mathcal{S}$ and  $G(\Vkt k'')\in\mathcal{S}^c$, respectively, using the notation of the above proof and 
$L=\sum_{i<j}L_{ij}$ denotes the framework's total length. 
Note that due to the division by the framework's volume $AL$, which is constant according to the assumed Poisson's ratio $\nu=1/2$, 
the snappability index $s(\Vkt k')$ can be interpreted as the change of the {\it elastic strain energy density} $U/(AL)$. 
Therefore $s(\Vkt k')$ is invariant with respect to scaling (taking into account $E_{ij}=1$; cf.\ Sec.\ \ref{geom_inter})  and enables the  
comparison of frameworks, which differ in the number of knots, the combinatorial structure and intrinsic metric. 
Note that the minimum of the obtained snappability indices over all undeformed realizations can be seen as the snappability index of the framework.   

Before we give the algorithm for computing the snappability index $s(\Vkt k)$ of an 
undeformed realization $G(\Vkt k)$, it should be noted 
that the concrete curve $(\Vkt k_t, U(\Vkt k_t))$  on the graph connecting %the local minimum 
$(\Vkt k, 0)$ and the saddle point $(\Vkt k'', U(\Vkt k''))$ does not play a role as long as 
the deformation energy $U_{ij}$ of each bar $e_{ij}$ is monotonic increasing 
with respect to the curve parameter $t$. 
This is due to the fact that along each curve of this possible set $\mathcal{C}$ of curves the same amount of mechanical work (namely the minimum work needed) 
is performed on the framework to reach the saddle \medskip point.

\noindent
{\bf Algorithm.} Given is an undeformed realization $G(\Vkt k)$. 
Let us assume that the unstable realization $G(\Vkt k'')\in\mathcal{S}^c$ yields the minimal function value $U(\Vkt k'')$. 
We consider the simplest possible path in $\RR^b$, namely the straight line segment from $\Vkt L$ to $\Vkt l''$ 
and parametrize it with respect to the time $t\in[0,1]$ yielding 
$\Vkt l_t:=\Vkt L+t(\Vkt l''-\Vkt L)$. 
This path corresponds to different 1-parametric deformations of realizations in $\EE^n$. 
If among these a deformation $G(\Vkt k_t)$ with the property 
\begin{equation}\label{property}
G(\Vkt k_t)\big|_{t = 0}=G(\Vkt k),\quad G(\Vkt k_t)\big|_{t = 1}=G(\Vkt k'') 
\end{equation}
%\begin{equation}\label{property}
%G(\Vkt k_0)=G(\Vkt k'),\quad G(\Vkt k_1)=G(\Vkt k'') 
%\end{equation}
exists ($\Rightarrow$ $(\Vkt k_t, U(\Vkt k_t))\in\mathcal{C}$), then the undeformed realization $G(\Vkt k)$ can be left over the unstable realization $G(\Vkt k'')$ and we get a value for $s(\Vkt k)$. 
Computationally the property (\ref{property}) can be checked e.g.\ by a parameter homotopy approach (e.g.\ Bertini; cf.\ \cite[Sec.\ 6]{bates}). 
If such a deformation does not exist then we redefine 
$\mathcal{S}^c$ as $\mathcal{S}^c\setminus\left\{ G(\Vkt k'')\right\}$ and run again the procedure explained in this paragraph until we end up with a value for $s(\Vkt k)$. In the case of $\mathcal{S}^c=\varnothing$ we 
set $s(\Vkt k)=\infty$.
%This loop either breaks down if we get a value for $s(\Vkt k')$ or 
%if one cannot snap out from the realizations $G(\Vkt k')$ 
%over a realizations of $\mathcal{S}^c$. In this case we set $s(\Vkt k')=\infty$.

\begin{remark}\label{op:bertini}
Note that our computational approach using Bertini does not recognize if the tracked path between the 
real starting point and real endpoint of the homotopy is entirely real. More generally it is an open problem to 
check if there exists at least one real curve $\in\mathcal{C}$, which corresponds 
to such a continuous real deformation.  \hfill $\diamond$
\end{remark}

After applying the mechanical minimum work needed to deform $G(\Vkt k')$ into  $G(\Vkt k'')$ as 
described in Theorem \ref{thm1}, the framework will 
relax according to the {\it minimum total potential energy principle}. Therefore this self-acting deformation will end up in 
a stable realization or it get stuck on the way to such a local minimum by reasons of reality. 
Note that a realization at the boarder of reality has to be infinitesimal flexible, as  a real solution of an algebraic set of equations 
can only change over into a complex one through a double root. 
This results in the following theorem:

\begin{theorem}\label{thm2}
A snap of a framework described in Theorem \ref{thm1} ends up in a realization $G(\Vkt k''')$ which is either undeformed or a deformed one with a shakiness. 
\end{theorem}

\begin{remark}
Note that $s(\Vkt k')$ is independent of the final snapping realization $G(\Vkt k''')$. \hfill $\diamond$ 
\end{remark}

\subsection{Pinned frameworks}
From the mechanical point of view it makes sense to fix a system to the ground by pinning a subset $\mathcal{P}$ of the knot set $\mathcal{K}$. 
The results of this paper also hold  for this scenario of so-called {\it pinned graphs} (or {\it grounded graphs}) due to their fundamental properties 
summarized in \cite[Sec.\ 2.1]{nixon}. One only has to keep in mind that bars between 
pinned knots cannot be deformed and that the equilibrium condition (\ref{equilibrium}) only has to hold for knots $\in \mathcal{K}\setminus\mathcal{P}$ 
as pinned knots can counterbalance any force. 
This becomes clear by studying a trivial example (comparison of a pinned and unpinned triangular framework) given in the Appendix (cf.\ Sec.\ \ref{ex:triangle}), 
where the computation of the snappability index 
is also demonstrated for a more sophisticated framework (pinned 3-legged planar parallel manipulator; cf.\ Sec.\ \ref{ex:parallel}).

\section{Conclusion and open problems}

The total elastic strain energy of the framework (based on a physical model for the deformation of bars using Hooke's law) 
serves as base for the presented snappability index and the theoretical results of Theorem \ref{thm1} and \ref{thm2}, which give further connections between 
shakiness and snapping beside the technique of averaging and deaveraging.

Note that our approach neglects the possibility of collision of bars during the framework's deformation. 
A further open problem is mentioned in Remark \ref{op:bertini}.

\begin{acknowledgement}
The research is supported by Grant No.\ P\,30855-N32 of the Austrian Science Fund FWF. 
Thanks to Hellmuth Stachel for constructive feedback on the final draft.
\end{acknowledgement}

\bibliographystyle{spmpsci}

\section{Appendix of examples}

\subsection{Pinned and unpinned triangular framework}\label{ex:triangle}

We study the triangle with vertices $K_1,K_2,K_3$ and edge lengths $L_{12}=10$, $L_{13}=7$ and $L_{23}=4$. 
In the first approach we consider an unpinned framework where we attach the fixed frame in the following way to the framework: 
The origin coincide with $K_1$ and the positive $x$-axis points into direction of $K_2$; i.e.\
\begin{equation}
\Vkt k_1=(0,0), \qquad \Vkt k_2=(k_{2,1},0), \qquad \Vkt k_3=(k_{3,1},k_{3,2}). 
\end{equation} 
In the second approach we pin $K_1$ and $K_2$ yielding
\begin{equation}
\Vkt k_1=(0,0), \qquad \Vkt k_2=(L_{12},0), \qquad \Vkt k_3=(k_{3,1},k_{3,2}). 
\end{equation}

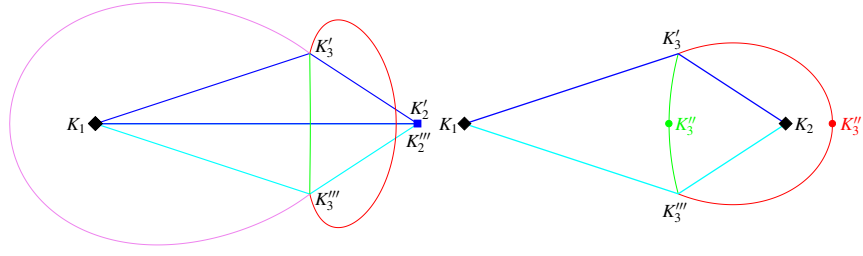
\begin{figure}[t]
\begin{center} 
\begin{overpic}
    [width=110mm]{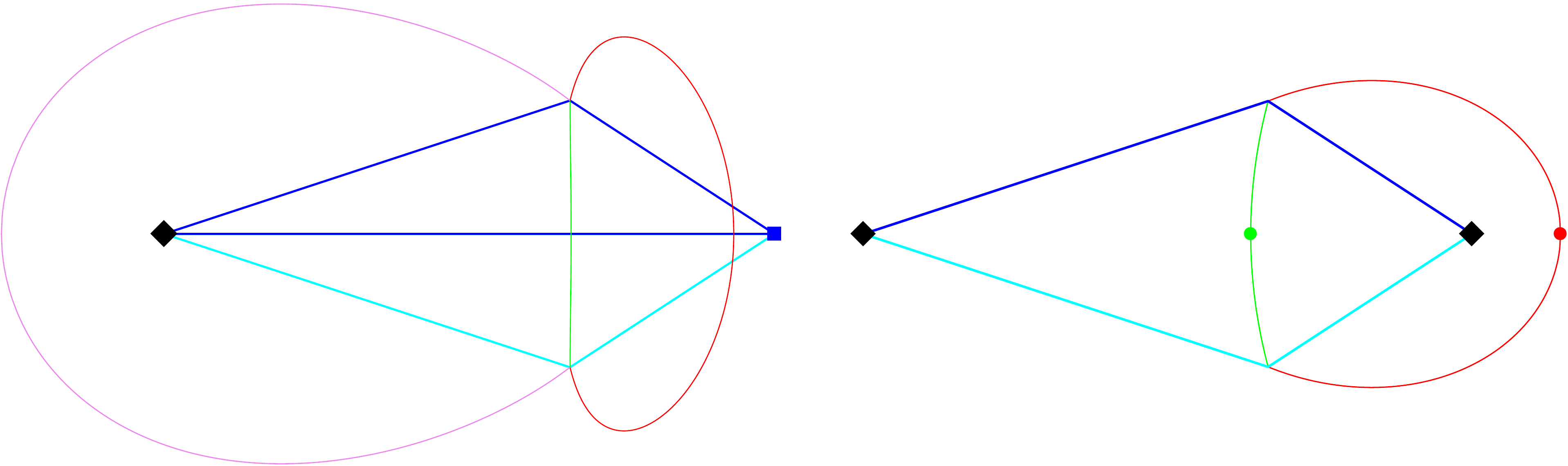}  
\begin{scriptsize}
\put(7,14){$K_1$}
\put(48.5,16){$K_2'$}
\put(48,12){$K_2'''$}
\put(37,23.5){$K_3'$}
\put(37,5){$K_3'''$}
\put(52,14){$K_1$}
\put(95,14){$K_2$}
\put(100.5,14){{\rot $K_3''$}}
\put(80.5,14){{\green $K_3''$}}
\put(79,24.5){$K_3'$}
\put(79,3.5){$K_3'''$}
\end{scriptsize} 		
  \end{overpic} 
\end{center} 
\caption{Trajectories of the vertex $K_3$ during the snapping deformations (three in the unpinned case; two in the pinned case) 
from the blue into the cyan realization, which correspond to straight line 
segments in the space $\RR^b$ of bar lengths. (right) Pinned case, where the pinned knots are indicated by the black diamonds. 
(left) Also in the unpinned case we can pin the first knot $K_1$ without loss of generality. During the snapping deformation the horizontal bar with endpoint 
$K_2$ (indicated by a blue square) changes its length, which is pointed out in Fig.\ \ref{fig:2}. 
}\label{fig:1}	
\end{figure}

\noindent
{\bf Unpinned case:} 
According to the computation of Sec.\ \ref{subsec:comp}, the set $\mathcal{S}$ consists of the following two 
undeformed realizations:
\begin{equation}\label{unpins}
{\blau (k_{2,1},k_{3,1},k_{3,2})=(10,\tfrac{133}{20},\tfrac{7}{20}\sqrt{39})}, \quad
{\cyan (k_{2,1},k_{3,1},k_{3,2})=(10,\tfrac{133}{20},-\tfrac{7}{20}\sqrt{39})},
\end{equation}
and the set $\mathcal{S}^c$ of the following three unstable realizations: 
\begin{align}
&{\rot (k_{2,1},k_{3,1},k_{3,2})=(\tfrac{20}{3},\tfrac{28}{3},0)}     &\quad &\Longrightarrow &\quad &{\rot (l_{12},l_{13},l_{23})=(\tfrac{20}{3},\tfrac{28}{3},\tfrac{8}{3})}, \\  
&{\violet (k_{2,1},k_{3,1},k_{3,2})=(\tfrac{80}{21},-\tfrac{8}{3},0)} &\quad &\Longrightarrow &\quad &{\violet (l_{12},l_{13},l_{23})=(\tfrac{80}{21},\tfrac{8}{3},\tfrac{136}{21})}, \\ 
&{\green (k_{2,1},k_{3,1},k_{3,2})=(\tfrac{220}{21},\tfrac{20}{3},0)} &\quad &\Longrightarrow &\quad &{\green (l_{12},l_{13},l_{23})=(\tfrac{220}{21},\tfrac{20}{3},\tfrac{80}{21})}. 
\end{align}
The two undeformed realizations and the trajectories of $K_3$ under the snapping deformations between them passing through the unstable realizations are illustrated 
in Fig.\ \ref{fig:1}(left) and Fig.\ \ref{fig:2}, respectively. Note that the snappability index 
$s(10,\tfrac{133}{20},\pm\tfrac{7}{20}\sqrt{39})$ equals $1/882$, which corresponds to the 
{\it elastic strain energy density} of the green unstable realization. The corresponding values for the red and violet unstable realizations are $49/882$  and $169/882$, \bigskip respectively.

\noindent
{\bf Pinned case:} We obtain the same set $\mathcal{S}$ as in the unpinned case (cf.\ Eq.\ (\ref{unpins})), but the 
set $\mathcal{S}^c$ consists only of the following two realizations:
\begin{align}
&{\rot (k_{3,1},k_{3,2})=(\tfrac{126}{11},0)}     &\quad &\Longrightarrow &\quad &{\rot (l_{13},l_{23})=(\tfrac{126}{11},\tfrac{16}{11})}, \\  
&{\green (k_{3,1},k_{3,2})=(\tfrac{70}{11},0)} &\quad &\Longrightarrow &\quad &{\green (l_{13},l_{23})=(\tfrac{70}{11},\tfrac{40}{11})}, 
\end{align}
which are illustrated in Fig.\ \ref{fig:1}(right) together with  the trajectories of $K_3$ under the two corresponding snapping deformations. 
For this trivial example the two saddle points of the graph of the {\it elastic strain energy density function} $U(k_{3,1},k_{3,2})/(AL)$ can even 
be visualized (see Fig.\ \ref{fig:3}). The function values of the red and green saddle points are $49/462$ and $1/462$, respectively, where the 
latter one equals the snappability index $s(\tfrac{133}{20},\pm\tfrac{7}{20}\sqrt{39})$ of the pinned triangular framework.

\begin{figure}[t]
\begin{overpic}
    [width=113mm]{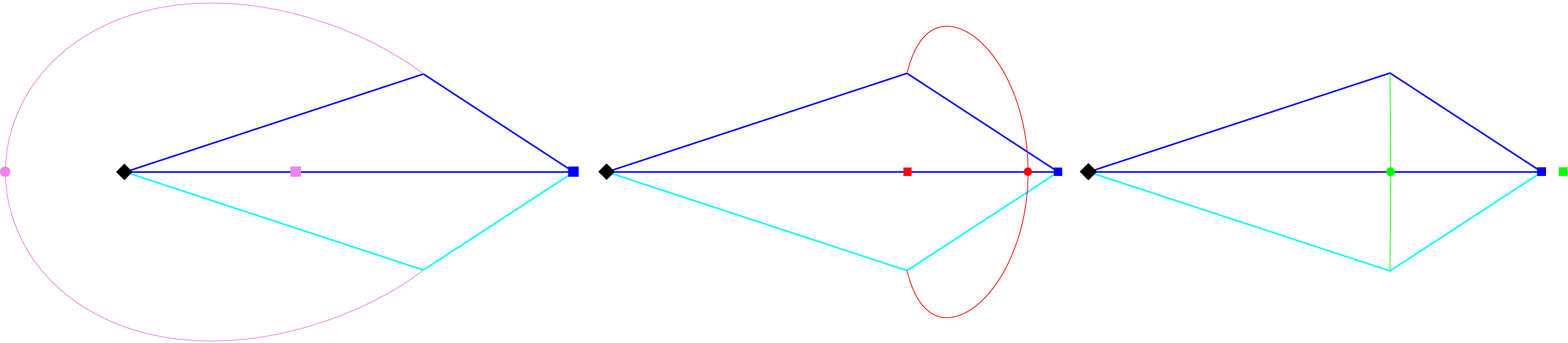}  
\begin{scriptsize}
\put(1,10.2){{\violet $K_3''$}}
\put(18,8.6){{\violet $K_2''$}}
\put(62.5,10.2){{\rot $K_3''$}}
\put(56.7,8.6){{\rot $K_2''$}}
\put(100.3,10.2){{\green $K_2''$}}
\put(88.8,8.6){{\green $K_3''$}}

\end{scriptsize} 			
  \end{overpic} 
\caption{Unstable realizations of the unpinned triangular framework: Each snapping deformation passes an unstable realization, which is 
displayed in violet (left), red (center) and green (right), respectively. The corresponding animations of these snapping deformations can be 
downloaded from {\tt www.dmg.tuwien.ac.at/nawratil/publications.html}.  
}\label{fig:2}	
\end{figure}

\begin{figure}[b]
\begin{center} 
\begin{overpic}
    [width=55mm]{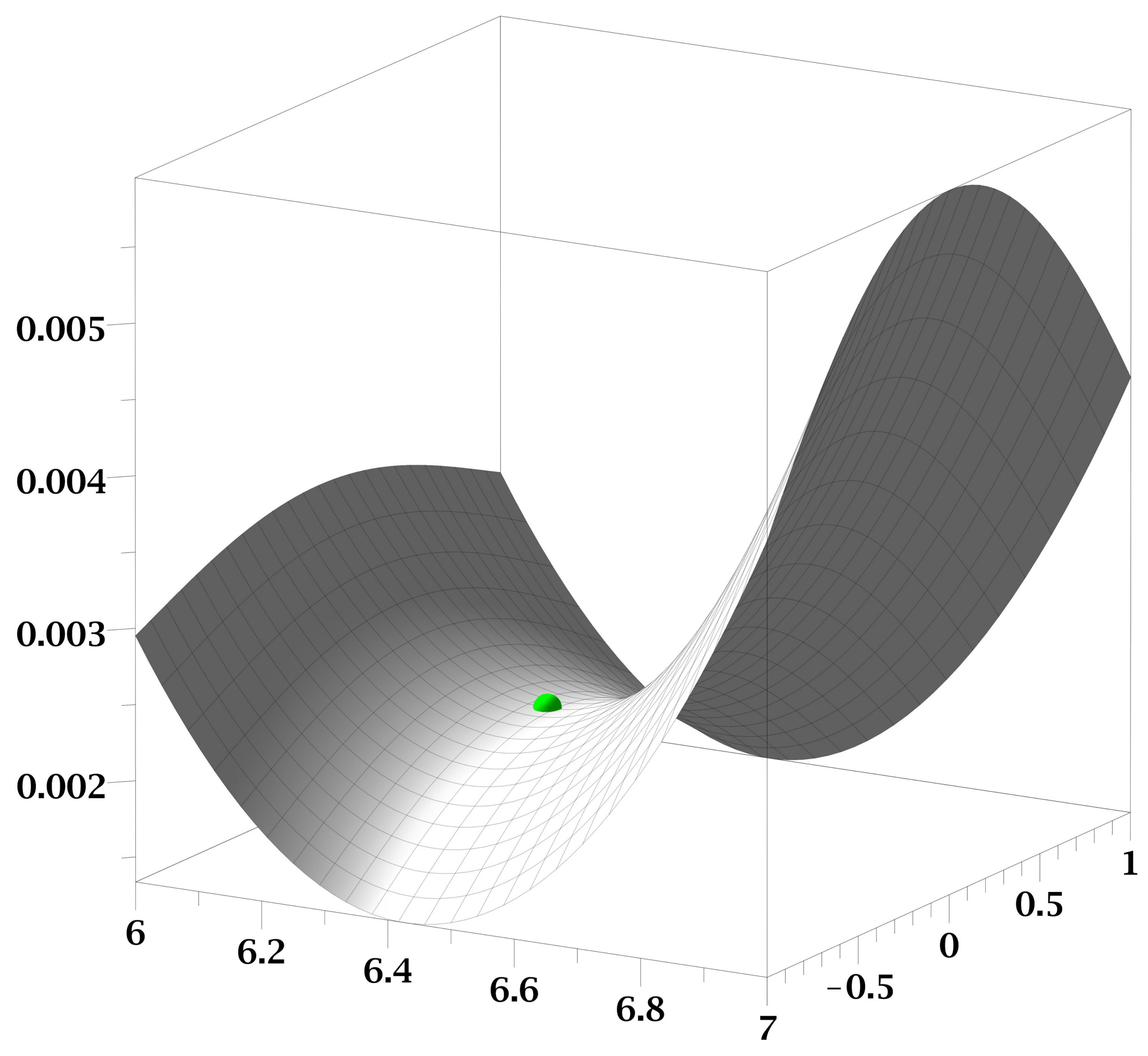}  
  \end{overpic} 
	\hfill
 \begin{overpic}
    [width=55mm]{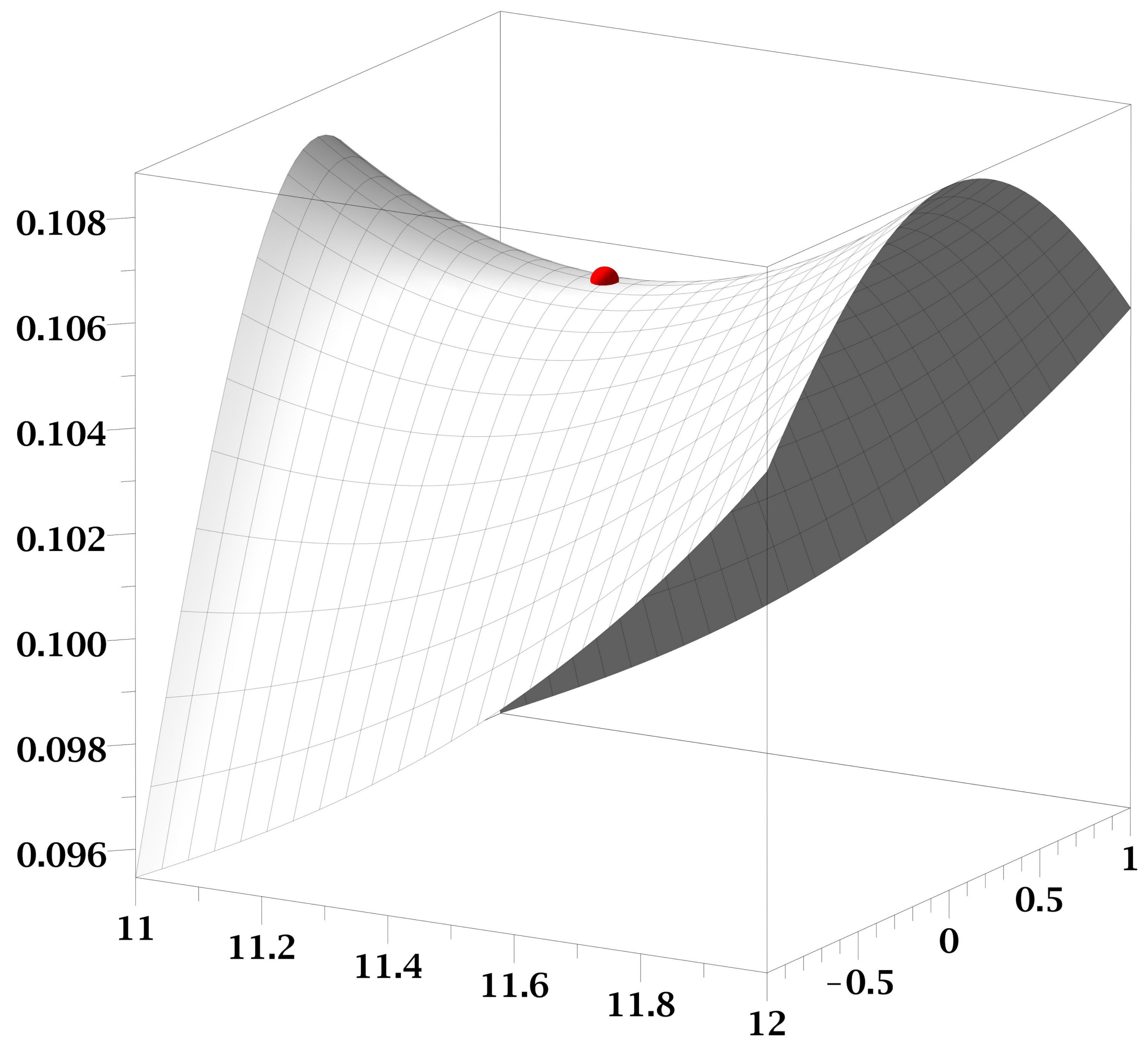}
  \end{overpic} 
\end{center} 
\caption{
Graph of the function $U(k_{3,1},k_{3,2})/(AL)$ of the pinned triangular framework: (left) $k_{3,1}\in[6;7]$ and $k_{3,2}\in[-1;1]$  
(right) $k_{3,1}\in[11;12]$ and $k_{3,2}\in[-1;1]$. The green and red saddle points correspond to the two unstable realization 
indicated in Fig.\ \ref{fig:1}(right).
}\label{fig:3}	
\end{figure}

\subsection{Pinned 3-legged planar parallel manipulator}\label{ex:parallel}

This more sophisticated framework (compared to the triangular one studied before) consists of six knots $K_1,\ldots ,K_6$, where 
the first three knots are pinned to the ground possessing the following coordinates:
\begin{equation}
\Vkt k_1=(0,0), \qquad \Vkt k_2=(3,0), \qquad \Vkt k_3=(2,1). 
\end{equation}
Each of these knots $K_i$ is connected by a bar, the so-called $i$-th leg, with one of the remaining three knots $K_{i+3}$ (for $i=1,2,3$) with
\begin{equation}
\Vkt k_4=(k_{4,1},k_{4,2}), \qquad \Vkt k_5=(k_{5,1},k_{5,2}), \qquad \Vkt k_6=(k_{6,1},k_{6,2}),
\end{equation} 
which form a joint-bar triangle (cf.\ Fig.\ \ref{fig:4}). The intrinsic metric of the framework is given by:
\begin{equation}\label{bar_bc}
(L_{14},L_{25},L_{36},L_{45},L_{46},L_{56})=(4,5,3,3,1,2)
\end{equation}
and implies that the joint-bar triangle $K_4,K_5,K_6$ degenerates as $L_{46}+L_{56}=L_{45}$ holds.
According to the computation of Sec.\ \ref{subsec:comp}, the set $\mathcal{S}$ consists of the following three stable
realizations:
\begin{equation}
\begin{split}
{\blau (\Vkt k_4,\Vkt k_5,\Vkt k_6)}&{\blau =((0.8876,3.9002),(-1.5278,2.1210),(0.0824,3.3071))}, \\
{\cyan (\Vkt k_4,\Vkt k_5,\Vkt k_6)}&{\cyan =((2.0771,3.4184),(4.8072,4.6619),(2.9871,3.8329))}, \\
{\magenta (\Vkt k_4,\Vkt k_5,\Vkt k_6)}&{\magenta =((2.9116,-2.4707),(0.3581,-3.8446),(1.8691,-2.2512))},
\end{split}
\end{equation}
where the first (blue) and second (cyan) are undeformed realizations and the third (magenta) is a deformed one (cf.\  Fig.\ \ref{fig:4}(left)). 
The {\it elastic strain energy density} of the latter realization equals $0.00219$. 

The set $\mathcal{S}^c$ contains $47$ unstable realizations where the one with the smallest 
{\it elastic strain energy density} value of $0.00061$ is given by:
\begin{equation}\label{real_green}
{\green (\Vkt k_4,\Vkt k_5,\Vkt k_6)=((3.2050,2.5883),(1.4895,4.8801),(3.1261,3.6410))},
\end{equation}
which is illustrated in Fig.\ \ref{fig:4}(right). The remaining 46 elements of $\mathcal{S}^c$ are given in Table \ref{tab1} for reasons of completeness.

\begin{remark}
The critical points of the elastic strain energy function where computed by Bertini based on the splitting of the variables into the following two groups:
\begin{equation*}
(k_{4,1},k_{4,2},k_{5,1},k_{5,2},k_{6,1},k_{6,2},q_{14},q_{25},q_{36},q_{45},q_{46},q_{56}), \phm (\lambda_{14},\lambda_{25},\lambda_{36},\lambda_{45},\lambda_{46},\lambda_{56}), 
\end{equation*}
which resulted in 59136 paths. In contrast the full homotopy yields 262144 paths. \hfill $\diamond$
\end{remark}

\begin{table}[t]
\begin{center}
\begin{footnotesize}
\begin{tabular}{|c|c|c|c|c|c||c|}
\hline
$k_{4,1}$	& $k_{5,1}$ & $k_{6,1}$ & $k_{4,2}$ & $k_{5,2}$ & $k_{6,2}$ & $U/(AL)$     \\ \hline\hline
%
%3.205 & 1.490 & 3.126 & 2.588 & 4.880 & 3.641 & 0.0006 \\ \hline %35
4.030 & 6.239 & 4.489 & -0.5010 & -2.720 & -1.352 & 0.0058 \\ \hline %17
-2.579 & -1.714 & -1.482 & -1.345 & 1.415 & -0.6494 & 0.0177 \\ \hline  %25
0.9358 & 2.609 & 0.3781 & 3.889 & 4.985 & 3.524 & 0.0185 \\ \hline %4
1.179 & -0.7525 & 1.823 & 3.822 & 3.304 & 3.995 & 0.0185 \\ \hline %5
2.778 & 2.941 & 2.750 & -2.770 & -4.796 & -2.078 & 0.0191 \\ \hline %41
-2.456 & -0.3583 & -1.199 & -1.645 & -3.345 & -1.294 & 0.0201 \\ \hline %11
-2.156 & 0.1224 & -1.160 & 3.506 & 2.084 & 3.090 & 0.0204 \\ \hline  %24
2.212 & 1.947 & 2.213 & 3.810 & 5.681 & 3.265 & 0.0290 \\ \hline %45
4.987 & 6.767 & 4.520 & 0.7908 & 2.181 & 0.5923 & 0.0315 \\ \hline %16
-0.0774 & 0.9562 & -0.3930 & 3.889 & 2.578 & 4.320 & 0.0504 \\ \hline %40
1.822 & -1.357 & 0.5219 & 0.7588 & 1.022 & 0.1156 & 0.0512 \\ \hline %39
3.841 & 1.058 & 3.144 & 0.5499 & 1.821 & -0.1909 & 0.0585 \\ \hline  %21
3.484 & 0.9825 & 3.365 & -0.9764 & 1.343 & -0.0521 & 0.0627 \\ \hline %33
3.717 & 1.094 & 3.193 & -0.7774 & 1.565 & -0.0846 & 0.0629 \\ \hline %31
-2.140 & -3.594 & -1.331 & -0.2606 & -0.5004 & 0.2326 & 0.0633 \\ \hline %6
-1.683 & -1.276 & -2.146 & -0.7357 & -3.000 & -0.7246 & 0.0685 \\ \hline %38
-1.739 & -1.943 & -1.977 & -0.8342 & 1.398 & -1.090 & 0.0709 \\ \hline %10
-1.903 & -1.030 & -0.4593 & 3.516 & 3.015 & 2.686 & 0.0740 \\ \hline %2 
-1.658 & -0.6953 & -0.0537 & 3.640 & 3.368 & 3.187 & 0.0741 \\ \hline %8
-1.773 & -1.128 & -0.6979 & 3.586 & 2.821 & 2.312 & 0.0741 \\ \hline %43
3.046 & 0.6959 & 3.677 & -0.2001 & 0.9677 & -0.3291 & 0.0744 \\ \hline %52
-1.042 & -0.5401 & -0.1441 & -3.448 & -2.650 & -1.908 & 0.0813 \\ \hline %3
1.789 & 2.145 & 2.256 & 3.171 & 4.195 & 4.650 & 0.0828 \\ \hline  %26
-2.847 & -1.875 & -1.217 & -1.478 & -0.9779 & -0.6382 & 0.0828 \\ \hline %53
-1.275 & -3.596 & -1.799 & -0.2986 & -0.5276 & -0.4576 & 0.0857 \\ \hline %9
2.583 & 3.471 & 3.561 & -2.189 & -3.079 & -3.230 & 0.1229 \\ \hline  %29
3.718 & 5.109 & 5.277 & -1.870 & -1.684 & -1.725 & 0.1232 \\ \hline %20
0.9849 & -1.243 & 1.121 & 0.4462 & 1.478 & 0.4831 & 0.1246 \\ \hline  %27
5.734 & 4.377 & 5.733 & -0.9031 & -0.2398 & 0.2446 & 0.1262 \\ \hline %36
5.809 & 4.223 & 5.263 & -0.3537 & -0.3530 & -1.229 & 0.1328 \\ \hline %44
6.161 & 4.209 & 5.361 & -0.4353 & -0.2030 & -0.4448 & 0.1364 \\ \hline %19
5.852 & 5.521 & 4.580 & -0.0429 & -0.1258 & -0.7442 & 0.1398 \\ \hline %37
6.070 & 5.767 & 4.694 & 0.2238 & 0.2303 & 0.3625 & 0.1410 \\ \hline  %23
5.568 & 4.914 & 4.844 & -0.1332 & -0.5524 & -1.413 & 0.1422 \\ \hline  %30
5.177 & 4.151 & 5.738 & -0.8735 & -0.4692 & -1.234 & 0.1434 \\ \hline %34
1.719 & 0.0472 & -0.2643 & 0.2274 & -0.0260 & -0.0516 & 0.1446 \\ \hline %47
2.555 & 4.055 & 1.465 & 1.798 & -0.9142 & 1.256 & 0.1590 \\ \hline %54
-1.332 & -1.942 & -2.359 & -0.6569 & -0.9574 & -1.164 & 0.1593 \\ \hline %42
0.7485 & 1.643 & -0.2211 & 0.0430 & 0.0195 & -0.5323 & 0.1894 \\ \hline %51
1.870 & 2.051 & 2.113 & 2.226 & 1.752 & 0.2477 & 0.1934 \\ \hline %32
1.590 & 4.243 & 1.968 & 0.8727 & -1.124 & 1.006 & 0.2044 \\ \hline  %22
-1.026 & 3.253 & 0.2169 & 0.1767 & 0.0083 & -0.2709 & 0.2254 \\ \hline %46
2.261 & 2.692 & 2.076 & 2.313 & 0.3841 & 0.8505 & 0.2281 \\ \hline %18
3.500 & 2.403 & 2.367 & 1.631 & 0.5951 & 0.6044 & 0.2340 \\ \hline %15
0.0375 & 0.1406 & 0.9122 & 0.0104 & 0.0386 & 0.4382 & 0.2580 \\ \hline %7
-0.3547 & 3.084 & -0.5894 & 0.0129 & -0.0005 & 0.0230 & 0.2839 \\ \hline  %28
\end{tabular}
\end{footnotesize}
\end{center}
\caption{The remaining 46 elements of $\mathcal{S}^c$ beside the one given in Eq.\ (\ref{real_green}) ordered 
with respect to the elastic strain energy density $U/(AL)$.}
\label{tab1}
\end{table}

The framework cannot snap out of the magenta realization by passing the green one, as the 
{\it elastic strain energy density} of the latter realization is lower. Therefore we consider in the space $\RR^6$ of edge lengths 
the straight line segment $\Vkt l_t$ between the point given in Eq.\ (\ref{bar_bc}) and the corresponding point of the green realization given by
\begin{equation}\label{bar_green}
{\green (l_{14},l_{25},l_{36},l_{45},l_{46},l_{56})=(4.1196, 5.1085, 2.8710, 2.8626, 2.0527, 1.0555)}.
\end{equation}

\begin{figure}[t]
\begin{flushright} 
\begin{overpic}
    [width=115mm]{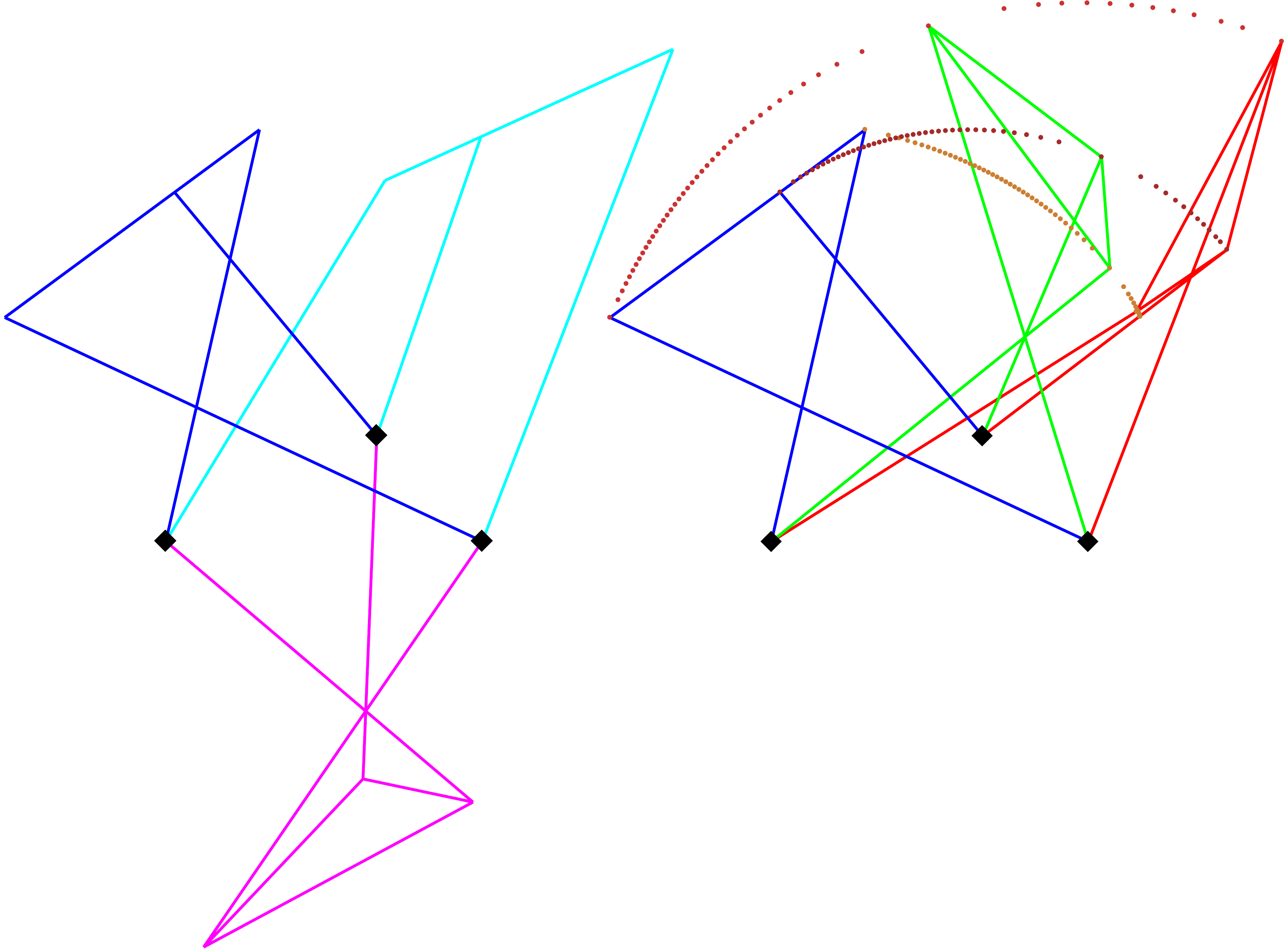} 
	\begin{scriptsize}
\put(14.5,31.5){{$K_1$}}
\put(39,31.5){{$K_2$}}
\put(25.8,39.5){{$K_3$}}
\put(19,64.8){{$K_4$}}
\put(-1.7,48){{$K_5$}}
\put(11.5,60){{$K_6$}}
\end{scriptsize} 			
  \end{overpic} 
\end{flushright} 
\caption{(left) Illustration of the three stable realizations: the blue and the cyan realization are undeformed and the magenta one is deformed.  
%with $l_{14}=3.8187$,   $l_{25}=4.6648$,   $l_{36}=3.2538$,   $l_{45}=2.8996$,   $l_{46}=1.0654$,  $l_{56}=2.1959$. 
The latter one is shaky as it is a deformed ($\Rightarrow$ non-zero self-stress) stable realization. 
This can also be seen by the fact that the three legs belong to a pencil of lines. 
Note that in this example also the undeformed realizations are shaky due to the degenerated triangle $K_4,K_5,K_6$. \newline
(right) Snapping deformation from the blue realization, passing through the green unstable realization and ending up at the 
red realization, which is on the boarder of reality. The line segment between the corresponding points of the 
blue and green realization (cf.\ Eqs.\ (\ref{bar_bc}) and (\ref{bar_green})) in the space $\RR^6$ of edge lengths is equally discretized (40 intervals).
The same holds for the  line segment between  the corresponding points of the green realization and the complex one given in Eq.\ (\ref{compl_real}). 
}\label{fig:4}	
\end{figure}
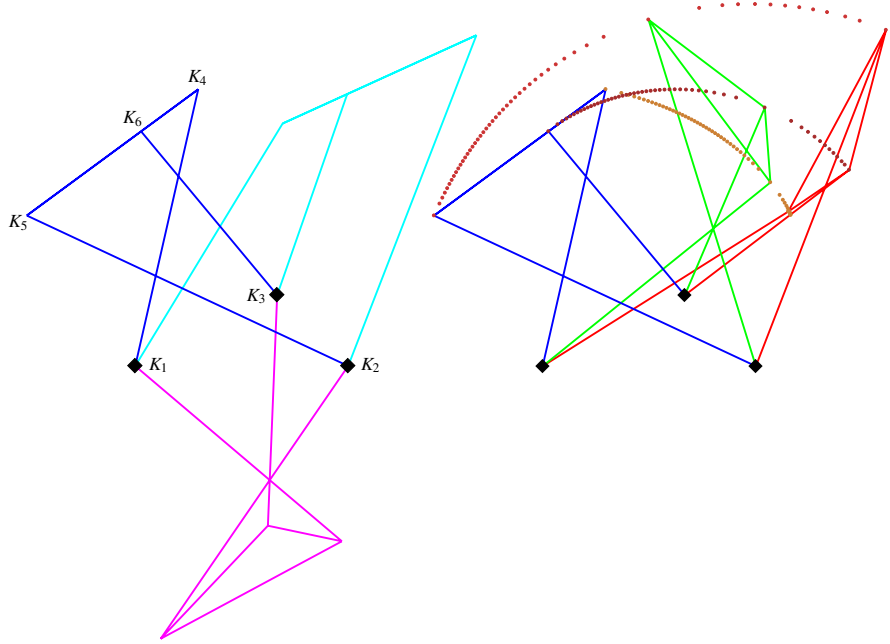

One of the corresponding 1-parametric deformations has the property that it connects the blue realization with the green realization (cf.\ Eq.\ (\ref{property})). 
A further corresponding deformation also ends up in the green realization. 
This second deformation does not start at the cyan realization, but in the following complex solution for an undeformed realization of the framework:
\begin{equation}\label{compl_real}
\begin{split}
(\Vkt k_4,\Vkt k_5,\Vkt k_6)=
((&3.8697-0.5179i,1.6591+1.2081i),(6.9159+0.1526i, \\ &3.1185-0.1916i),(4.8851-0.2944i,2.1456+0.7415i)). 
\end{split}
\end{equation}
Therefore  the framework will relax from the green realization towards this complex solution. 
The realization, where this 1-parametric deformation hits the boarder of reality, is illustrated by the red shaky realization in Fig.\ \ref{fig:4}(right).

\end{document}